\documentclass[11pt,amssymb,amsfont,a4paper]{article}
\usepackage{latexsym,amssymb,amsmath,amsthm,color}
\usepackage{geometry}
\usepackage{url}
\usepackage[utf8]{inputenc}
\usepackage{authblk}
\usepackage{bbm}
\theoremstyle{definition}
\newtheorem{definition}{Definition}
\newtheorem{example}{Example}

\theoremstyle{plain}
\newtheorem{theorem}{Theorem}
\newtheorem{proposition}{Proposition}
\newtheorem{lemma}{Lemma}

\newtheorem{corollary}{Corollary}

\title{Universal secure rank-metric coding schemes with optimal communication overheads \footnote{Parts of this paper have been accepted for presentation at the IEEE International Symposium on Information Theory, Aachen, Germany, June 2017. \cite{rankefficient-isit}}}
\author{Umberto Mart{\'i}nez-Pe\~{n}as \thanks{umberto@math.aau.dk}}
\affil{Department of Mathematical Sciences, Aalborg University, Denmark}
\date{}

\begin{document}

\maketitle

\begin{abstract}
We study the problem of reducing the communication overhead from a noisy wire-tap channel or storage system where data is encoded as a matrix, when more columns (or their linear combinations) are available. We present its applications to reducing communication overheads in universal secure linear network coding and secure distributed storage with crisscross errors and erasures and in the presence of a wire-tapper. Our main contribution is a method to transform coding schemes based on linear rank-metric codes, with certain properties, to schemes with lower communication overheads. By applying this method to pairs of Gabidulin codes, we obtain coding schemes with optimal information rate with respect to their security and rank error correction capability, and with universally optimal communication overheads, when $ n \leq m $, being $ n $ and $ m $ the number of columns and number of rows, respectively. Moreover, our method can be applied to other families of maximum rank distance codes when $ n > m $. The downside of the method is generally expanding the packet length, but some practical instances come at no cost.


\end{abstract}

\section{Introduction} \label{sec intro}

Universal secure linear network coding with errors and erasures was first studied in \cite{silva-universal}, where rank-metric coding schemes were proposed to protect messages sent over a linearly coded network from link errors, erasures and information leakage to a wire-tapper. Similarly, rank-metric codes have been applied to storage systems where data is stored as a matrix and where errors and erasures affect several rows and/or columns, also called crisscross errors and erasures \cite{roth}. These errors and erasures have been recently motivated by correlated and mixed failures in distributed storage systems where data is stored in several data centers (columns), which in turn store several blocks of data (rows). See \cite{swanan}.

In this paper, we study how to reduce the communication overhead from such a noisy wire-tap channel or storage system to the receiver, when more columns, or their linear combinations, are available: Less ingoing links to the receiver fail in the network case, or more data centers are available and contacted in the distributed storage case. As it has been noticed in secret sharing in the literature \cite{rawad, efficient, wangwong}, which corresponds to Hamming-metric erasure-correction and security, if more pieces of data (columns in our case) are available, they can be preprocessed via subpacketization so that the overall transmitted information from the channel or storage system to the receiver is reduced. 

A similar concept of subpacketization has been recently developed for Reed-Solomon codes in \cite{guruswami}. In another direction, coding schemes recovering part of the encoded data (a node in a storage system, for instance), with respect to the Hamming metric, have already been studied, giving rise to \textit{regenerating codes} \cite{dimakis1, dimakis2, rashmi}, which reduce communication bandwidth, and \textit{locally repairable codes} \cite{localitysymbols, pyramidcodes, tamobarg}, which reduce the number of contacted nodes. The latter codes have been recently extended to the rank metric in \cite{swanan}. In contrast, our aim is to recover the whole uncoded data while reducing the communication bandwidth, as in \cite{rawad, efficient, wangwong}, but with respect to the rank metric and, as a consequence, with respect to the crisscross metric.

We illustrate and motivate the problem with a pair of examples. The details of the constructions will be given in Subsection \ref{subsec using gabidulin}.

\begin{example} \label{example 1}
Consider a linearly coded network, as in \cite[Sec. VII-A]{silva-universal}, over a finite field of size $ q = 256 $ ($ 8 $-bit symbols), with packet length $ m = 2048 $, number of outgoing links from the source $ n = 40 $, at least $ N \geq n $ ingoing links to the sink, and where $ \mu \leq 8 $ links may be wire-tapped and $ \rho \leq 16 $ ingoing links to the sink may fail.

In \cite[Th. 11]{silva-universal}, a coding scheme is given with optimal information rate $ 16 / 40 $, able to correct the given number of erasures and secure under the given number of observations over such network, independently of its inner code (\textit{universally}). The overall communication overhead from the last ingoing links to the sink is of $ 8 $ packets: The source wants to transmit $ 16 $ uncoded packets and the sink receives $ 24 $ encoded packets. 

Thanks to Theorem \ref{theorem using gabidulin} and dividing each packet into $ 32 $ subpackets of length $ 64 $ each, we will obtain a coding scheme with the same parameters, but such that the overall communication overhead at the ingoing links to the sink is of $ 4 $ packets (the minimum possible) if none of them fail (only $ 20 $ packets are received by the sink).
\end{example}

\begin{example} \label{example 2}
Let again $ n = 40 $ and $ m = 2048 $, and consider a distributed storage system where data is stored as an $ m \times n $ matrix over the same finite field ($ q = 256 $), where each column corresponds to a data center that stores $ n $ symbols over $ \mathbb{F}_q $, that is, $ 40 $ $ 8 $-bit symbols. Assume that $ \rho $ data centers may fail or not be available, errors occur along $ t $ rows and/or columns due to certain correlations, and a wire-tapper eavesdrops $ \mu $ data centers. Assume also that $ \rho + 2t \leq 16 $ and $ \mu \leq 8 $.

As in the previous example, the use of a pair of maximum rank distance codes allows to obtain the desired reliability and security while achieving the optimal information rate $ 16 / 40 $ (see \cite{roth}), with a communication overhead of $ 8 $ packets from the contacted data centers to the receiver. Again, in this work we obtain a coding scheme with the same parameters but where the communication overhead is reduced to $ 4 $ packets (the minimum possible) if no errors occur and all data centers are available and contacted.
\end{example}

The paper is organized as follows: In Section \ref{sec info theory setting}, we establish the information-theoretical setting, defining coherent linearized noisy wire-tap channels, which we take from \cite{silva-universal}, and we establish a method of subpacketization that allows to use linear codes over the extension field. In Section \ref{sec comm efficiency for lin channels}, we define communication overheads for these linearized channels and give lower bounds on these parameters similar to those in \cite{efficient}. In Section \ref{sec general construction}, we give the main contribution of this paper, which is a general method to transform coding schemes based on pairs of linear rank-metric codes, with certain properties, into coding schemes with lower communication overheads. In Section \ref{sec optimal schemes}, we apply Gabidulin codes \cite{gabidulin, roth} to obtain coding schemes with optimal information rates and communication overheads for $ n \leq m $, which can be seen as a rank-metric analog of the constructions in \cite{rawad, efficient}. However, our method allows us to correct errors, and not only erasures as in the secret sharing case \cite{rawad, efficient}, and can be applied to other families of maximum rank distance codes, such as those in \cite{reducible} for $ n > m $. Finally, in Section \ref{sec applications}, we discuss the applications in universal secure linear network coding and secure distributed storage with crisscross errors and erasures.

\subsection*{Notation}

Throughout the paper, we fix a prime power $ q $ and positive integers $ m $, $ n $, $ N $, $ \alpha $, $ \ell $, $ t $, $ \rho $ and $ \mu $. We denote by $ \mathbb{F}_q $ the finite field with $ q $ elements, $ \mathbb{F}_q^n $ denotes the set of row vectors of length $ n $ over $ \mathbb{F}_q $, and $ \mathbb{F}_q^{m \times n} $ denotes the set of $ m \times n $ matrices over $ \mathbb{F}_q $. In this paper, a code is a subset of either $ \mathbb{F}_q^n $ or $ \mathbb{F}_q^{m \times n} $, whose linearity properties are specified in each case. For an $ \mathbb{F}_{q^m} $-linear code $ \mathcal{C} \subseteq \mathbb{F}_{q^m}^n $, we will denote by $ \mathcal{C}^\perp $ its dual code with respect to the usual $ \mathbb{F}_{q^m} $-bilinear inner product. We also use the notation $ [n] = \{ 1,2, \ldots, n \} $ and $ [m,n] = \{ m, m+1, \ldots, n \} $ whenever $ m \leq n $, and we denote by $ H(X) $, $ H(X \mid Y) $ and $ I(X;Y) $ the entropy, conditional entropy and mutual information of the random variables $ X $ and $ Y $, respectively (see \cite{cover}), where logarithms will always be taken with base $ q $.

\section{Information-theoretical setting and preliminaries} \label{sec info theory setting}

\subsection{Coherent linearized channels and coset coding schemes}

We will consider the secret message $ S $ to be a uniform random variable in $ \mathcal{S} = \mathbb{F}_q^{\alpha m \times \ell} $, and we will consider noisy wire-tap channels (which can also be thought of as distributed storage systems) as given in \cite{silva-universal}:

\begin{definition}[\textbf{Coherent linearized channel \cite{silva-universal}}] \label{def lin channel}
We define a coherent linearized noisy wire-tap channel with $ t $ errors, $ \rho $ erasures with erasure matrix $ A \in \mathbb{F}_q^{N \times n} $ of rank at least $ n-\rho $, and $ \mu $ observations as a channel with input a variable $ X \in \mathcal{X} = \mathbb{F}_q^{\alpha m \times n} $, output to the receiver $ Y \in \mathcal{Y} = \mathbb{F}_q^{\alpha m \times N} $, and output to the eavesdropper $ W \in \mathcal{W} = \mathbb{F}_q^{\alpha m \times \mu} $, together with a conditional probability distribution $ P(Y,W | X) $ such that
\begin{equation*}
\begin{split}
\mathcal{Y}_X = \{ Y \in \mathbb{F}_q^{\alpha m \times N} \mid & Y = XA^T + E, \\
& E \in \mathbb{F}_q^{\alpha m \times N}, {\rm Rk}(E) \leq t \}, \\
\mathcal{W}_X = \{ W \in \mathbb{F}_q^{\alpha m \times \mu} \mid & W = XB^T, B \in \mathbb{F}_q^{\mu \times n} \},
\end{split}
\end{equation*}
where $ \mathcal{Y}_X = \{ Y \in \mathcal{Y} \mid P(Y|X) > 0 \} $ and $ \mathcal{W}_X = \{ W \in \mathcal{W} \mid P(W|X) > 0 \} $, for a given $ X \in \mathcal{X} $. 
\end{definition}

In \cite{silva-universal}, it is shown that a linearly coded network over $ \mathbb{F}_q $ with link errors, erasures and information leakage, and where the last coding coefficients are known to the receiver, can be modelled as a coherent linearized noisy wire-tap channel. We will focus on this scenario and discuss how to translate the results to the distributed storage scenario with crisscross errors and erasures in Subsection \ref{subsec distributed}, since the latter can be seen as a simpler case.

As encoders, we consider coset coding schemes as in \cite[Def. 7]{rgrw}, which are a particular case of those in \cite{silva-universal}.

\begin{definition}[\textbf{Coset coding schemes \cite{rgrw}}]
A coset coding scheme over the field $ \mathbb{F}_q $ with secret message set $ \mathcal{S} = \mathbb{F}_q^{\alpha m \times \ell} $ and coded message set $ \mathcal{X} = \mathbb{F}_q^{\alpha m \times n} $ is a randomized function
$$ F : \mathbb{F}_q^{\alpha m \times \ell} \longrightarrow \mathbb{F}_q^{\alpha m \times n}, $$
where, for every $ S \in \mathbb{F}_q^{\alpha m \times \ell} $, $ C = F(S) $ is the uniform random variable over a set $ \mathcal{C}_S \subseteq \mathbb{F}_q^{\alpha m \times n} $. To allow correct decoding, we also assume that $ \mathcal{C}_S \cap \mathcal{C}_T = \varnothing $ if $ S \neq T $. Finally, we define the information rate of the scheme as
\begin{equation}
R = \frac{\log_q(\# \mathcal{S})}{\log_q(\# \mathcal{X})} = \frac{\alpha m \ell}{\alpha m n} = \frac{\ell}{n}.
\label{eq def info rate}
\end{equation}
\end{definition}

In linear network coding, universal reliability and security means correcting a number of link errors and erasures and being secure under a number of link observations, independently of the network inner code. This leads in \cite{silva-universal} to the following definition:

\begin{definition}[\textbf{Universal schemes \cite{silva-universal}}] \label{def universal schemes}
We say that the coset coding scheme $ F : \mathbb{F}_q^{\alpha m \times \ell} \longrightarrow \mathbb{F}_q^{\alpha m \times n} $ is:
\begin{enumerate}
\item
Universally $ t $-error and $ \rho $-erasure-correcting if, for every coherent linearized channel with $ t $ errors, $ \rho $ erasures and erasure matrix $ A \in \mathbb{F}_q^{N \times n} $, there exists a decoding function $ D_A : \mathcal{Y} \longrightarrow \mathcal{S} $ such that
$$ D_A(Y) = S, $$
for all $ Y \in \bigcup_{X \in \mathcal{C}_S} \mathcal{Y}_X $ and all $ S \in \mathcal{S} $.
\item
Universally secure under $ \mu $ observations if, for every coherent linearized channel with $ \mu $ observations, it holds that 
$$ H(S | W) = H(S), $$
or equivalently $ I(S;W) = 0 $, for all $ W \in \bigcup_{X \in \mathcal{C}_S} \mathcal{W}_X $ and all $ S \in \mathcal{S} $.
\end{enumerate}
\end{definition}

\subsection{Using linear codes over the extension field}

In what follows, we will make use of codes that are linear over the extension field $ \mathbb{F}_{q^m} $. To that end, we need to see how to identify matrices in $ \mathbb{F}_q^{\alpha m \times n} $ with matrices in $ \mathbb{F}_{q^m}^{\alpha \times n} $:

\begin{definition} \label{def phi map}
Fix a basis $ \gamma_1, \gamma_2, \ldots, \gamma_m $ of $ \mathbb{F}_{q^m} $ as a vector space over $ \mathbb{F}_q $, and define the map 
\begin{equation}
\varphi_n : \mathbb{F}_{q^m}^{\alpha \times n} \longrightarrow \mathbb{F}_q^{\alpha m \times n}
\label{eq def phi map}
\end{equation}
as follows: Given a matrix $ C \in \mathbb{F}_{q^m}^{\alpha \times n} $ with entries $ c_{i,j} \in \mathbb{F}_{q^m} $, for $ i = 1,2, \ldots, \alpha $ and $ j = 1,2, \ldots, n $, we define $ \varphi_n(C) $ as the unique $ \alpha m \times n $ matrix with coefficients $ d_{l,j} \in \mathbb{F}_q $, for $ l = 1,2, \ldots, \alpha m $ and $ j = 1,2, \ldots, n $, such that
$$ c_{i,j} = \sum_{u = 1}^m d_{(i-1)m + u,j} \gamma_u, $$
for $ i = 1,2, \ldots, \alpha $ and $ j = 1,2, \ldots, n $. Finally, we define the rank over $ \mathbb{F}_q $ of a matrix $ E \in \mathbb{F}_{q^m}^{\alpha \times n} $ as the rank over $ \mathbb{F}_q $ of the matrix $ \varphi_n(E) \in \mathbb{F}_q^{\alpha m \times n} $, and we denote it by $ {\rm Rk}_q(E) $.
\end{definition}

The key result is that the effect of coherent linearized noisy wire-tap channels in Definition \ref{def lin channel} remains unchanged by the map $ \varphi_n $, as we will now see:

\begin{lemma}
Let $ C \in \mathbb{F}_{q^m}^{\alpha \times n} $, $ A \in \mathbb{F}_q^{N \times n} $ and $ E \in \mathbb{F}_{q^m}^{\alpha \times N} $. It holds that
$$ \varphi_N \left( C A^T + E \right) = \varphi_n(C) A^T + \varphi_N(E), $$
and $ {\rm Rk}_q(E) = {\rm Rk}(\varphi_N(E)) $ by definition.
\end{lemma}
\begin{proof}
The additive property of $ \varphi_n $ is clear from the definition, so we may assume that $ E = 0 $. Denote the entries of $ C $ and $ \varphi_n(C) $ as in Definition \ref{def phi map}, and let $ a_{v,j} $, $ \widetilde{c}_{i,j} $ and $ \widetilde{d}_{l,j} $ be the entries of $ A $, $ CA^T $ and $ \varphi_N \left( CA^T \right) $, respectively, for $ v = 1,2, \ldots, N $, $ i = 1,2, \ldots, \alpha $, $ l = 1,2, \ldots, \alpha m $ and $ j = 1,2, \ldots, n $. It holds that
$$ \widetilde{c}_{i,j} = \sum_{v = 1}^n c_{i,v} a_{j,v} = \sum_{v=1}^n \left( \sum_{u = 1}^m d_{(i-1)m + u, v} \gamma_u \right) a_{j,v} $$
$$ = \sum_{u = 1}^m \left( \sum_{v=1}^n d_{(i-1)m + u, v} a_{j,v} \right) \gamma_u, $$
but it also holds that
$$ \widetilde{c}_{i,j} = \sum_{u = 1}^m \widetilde{d}_{(i-1)m + u,j} \gamma_u, $$
for $ i = 1,2, \ldots, \alpha $ and $ j = 1,2, \ldots, n $. Since $ \gamma_1, \gamma_2, \ldots, \gamma_m $ is a basis of $ \mathbb{F}_{q^m} $ over $ \mathbb{F}_q $ and $ \sum_{v=1}^n d_{(i-1)m + u, v} a_{j,v} \in \mathbb{F}_q $, for $ i = 1,2, \ldots, \alpha $, and $ j = 1,2, \ldots, n $, we conclude that
$$ \widetilde{d}_{(i-1)m + u,j} = \sum_{v=1}^n d_{(i-1)m + u, v} a_{j,v}, $$
for $ i = 1,2, \ldots, \alpha $, $ u = 1,2, \ldots, m $ and $ j = 1,2, \ldots, n $, which means that $ \varphi_N \left( CA^T \right) = \varphi_n(C) A^T $, and the result follows.
\end{proof}

Hence we may identify the sets $ \mathbb{F}_q^{\alpha m \times n} $ and $ \mathbb{F}_{q^m}^{\alpha \times n} $, seen as $ \mathbb{F}_q $-linear vector spaces together with the metric given by the rank and the function $ {\rm Rk}_q $, respectively. We will do this repeatedly throughout the paper.

To conclude the section, we recall the construction of coset coding schemes in \cite[Def. 4]{rgrw} based on pairs of $ \mathbb{F}_{q^m} $-linear codes with $ \alpha = 1 $ (no subpacketization). 

\begin{definition} [\textbf{Nested coset coding schemes \cite{rgrw}}] \label{def nested schemes}
A nested coset coding scheme (with $ \alpha = 1 $) is a coset coding scheme such that $ \mathcal{C}_S = \varphi(S) + \mathcal{C}_2 $, where $ \mathcal{C}_2 \subsetneqq \mathcal{C}_1 \subseteq \mathbb{F}_{q^m}^n $ are $ \mathbb{F}_{q^m} $-linear codes and $ \varphi : \mathbb{F}_{q^m}^{\ell} \longrightarrow \mathcal{W} $ is a vector space isomorphism over $ \mathbb{F}_{q^m} $, for an $ \mathbb{F}_{q^m} $-linear space $ \mathcal{W} \subseteq \mathbb{F}_{q^m}^n $ such that $ \mathcal{C}_1 = \mathcal{C}_2 \oplus \mathcal{W} $, where $ \oplus $ denotes the direct sum of vector spaces.
\end{definition}

To measure the reliability and security of these coding schemes, we need the concept of \textit{relative minimum rank distance}, which is a particular case of \cite[Def. 2]{rgrw}:

\begin{definition} [\textbf{Relative minimum rank distance \cite{rgrw}}]
Given $ \mathbb{F}_{q^m} $-linear codes $ \mathcal{C}_2 \subsetneqq \mathcal{C}_1 \subseteq \mathbb{F}_{q^m}^n $, we define their relative minimum rank distance as
$$ d_R \left( \mathcal{C}_1, \mathcal{C}_2 \right) = \min \left\lbrace {\rm Rk}_q(\mathbf{e}) \mid \mathbf{e} \in \mathcal{C}_1, \mathbf{e} \notin \mathcal{C}_2 \right\rbrace . $$
The minimum rank distance of a single code $ \mathcal{C} \subseteq \mathbb{F}_{q^m}^n $ is defined as $ d_R(\mathcal{C}) = d_R(\mathcal{C}, \{ \mathbf{0} \}) $.
\end{definition} 

The next result, which follows directly from \cite[Cor. 5 and Th. 4]{rgrw}, gives the mentioned reliability and security performance of nested coset coding schemes. Recall that we denote by $ \mathcal{C}^\perp $ the dual of an $ \mathbb{F}_{q^m} $-linear code $ \mathcal{C} \subseteq \mathbb{F}_{q^m}^n $ with respect to the usual $ \mathbb{F}_{q^m} $-bilinear inner product in $ \mathbb{F}_{q^m}^n $.

\begin{lemma} [\textbf{\cite{rgrw}}] \label{lemma correction capability}
Given $ \mathbb{F}_{q^m} $-linear codes $ \mathcal{C}_2 \subsetneqq \mathcal{C}_1 \subseteq \mathbb{F}_{q^m}^n $, the nested coset coding scheme in Definition \ref{def nested schemes} is universally $ t $-error and $ \rho $-erasure-correcting if, and only if, $ 2t + \rho < d_R(\mathcal{C}_1, \mathcal{C}_2) $, and is universally secure under $ \mu $ observations if, and only if, $ \mu < d_R \left( \mathcal{C}_2^\perp, \mathcal{C}_1^\perp \right) $.
\end{lemma}

Observe that $ d_R(\mathcal{C}_1) \leq d_R(\mathcal{C}_1, \mathcal{C}_2) $ and $ d_R \left( \mathcal{C}_2^\perp \right) \leq d_R \left( \mathcal{C}_2^\perp, \mathcal{C}_1^\perp \right) $, hence the minimum rank distances of $ \mathcal{C}_1 $ and $ \mathcal{C}_2^\perp $ give sufficient conditions on the number of correctable errors and erasures and on the number of links that may be wire-tapped without information leakage, respectively.

\section{Communication overheads in coherent linearized channels} \label{sec comm efficiency for lin channels}

In this section we formalize how, as in communication efficient secret sharing \cite{rawad, efficient, wangwong}, if a coset coding scheme is able to correct $ t $ errors and $ \rho $ erasures, but $ d > n - \rho $ pieces of information are available (the rank of $ A $ is at least $ d $), then we may reduce the communication overhead from the channel to the receiver by making use of the additional $ d - n + \rho > 0 $ linearly independent rows of $ A $. Observe that only erasures, and not errors, are considered in the Hamming analog described in \cite{rawad, efficient, wangwong}. 

Let $ F : \mathbb{F}_q^{\alpha m \times \ell} \longrightarrow \mathbb{F}_q^{\alpha m \times n} $ be a coset coding scheme, let $ A \in \mathbb{F}_q^{d \times n} $ be of rank $ d $ (if $ A \in \mathbb{F}_q^{N \times n} $ has rank $ d < N $, we may delete or ignore linearly dependent rows), let $ \mathbf{a}_1, \mathbf{a}_2, \ldots, \mathbf{a}_d \in \mathbb{F}_q^n $ be its rows, and let $ E_{A,i} : \mathbb{F}_q^{\alpha m} \longrightarrow \mathbb{F}_q^{\beta_i m} $ be preprocessing functions, where $ 1 \leq \beta_i \leq \alpha $, for $ i = 1,2, \ldots, d $. We define their correction capability with respect to $ F $ as follows:

\begin{definition} \label{def preprocessing f correction}
For a full-rank matrix $ A \in \mathbb{F}_q^{d \times n} $, the preprocessing functions $ E_{A,i} : \mathbb{F}_q^{\alpha m} \longrightarrow \mathbb{F}_q^{\beta_i m} $, for $ i = 1,2, \ldots, d $, are $ t $-error-correcting with respect to the coset coding scheme $ F : \mathbb{F}_q^{\alpha m \times \ell} \longrightarrow \mathbb{F}_q^{\alpha m \times n} $ if there exists a decoding function $ D_A : \prod_{i = 1}^d \mathbb{F}_q^{\beta_i m} \longrightarrow \mathbb{F}_q^{\alpha m \times \ell} $ such that
\begin{equation} 
 D_A \left( \left( E_{A,i} \left( C \mathbf{a}_i^T + \mathbf{e}_i \right) \right) _{i = 1}^d \right) = S, 
\label{eq condition preprocessing functions 1}
\end{equation}
for all $ C \in \mathcal{C}_S $, all $ S \in \mathbb{F}_q^{\alpha m \times \ell} $ and all error matrices $ E \in \mathbb{F}_q^{\alpha m \times d} $ of rank at most $ t $ with columns $ \mathbf{e}_1, \mathbf{e}_2, \ldots, \mathbf{e}_d \in \mathbb{F}_q^{\alpha m} $.
\end{definition}

We define then the decoding bandwidth and communication overhead as $ q $-analogs of those in \cite[Def. 2]{efficient}:

\begin{definition} [\textbf{Decoding bandwidth and communication overhead}] \label{def decoding and comm overhead}
For a full-rank matrix $ A \in \mathbb{F}_q^{d \times n} $ and functions $ E_{A,i} : \mathbb{F}_q^{\alpha m} \longrightarrow \mathbb{F}_q^{\beta_i m} $, for $ i = 1,2, \ldots, d $, we define their decoding bandwidth and communication overhead, respectively, as
$$ {\rm DB}(A) = \frac{\sum_{i = 1}^d \beta_i m}{\alpha m} = \sum_{i = 1}^d \frac{\beta_i}{\alpha} \quad \textrm{and} \quad {\rm CO}(A) = \sum_{i = 1}^d \frac{\beta_i}{\alpha} - \ell. $$
\end{definition}

Thus, if a \textit{packet} is a vector in $ \mathbb{F}_q^{\alpha m} $, then the decoding bandwidth is the amount (which need not be an integer due to the subpacketization) of packets that the receiver obtains, or needs to obtain, from the channel, and the communication overhead is the difference with respect to the original number of uncoded packets.

Observe that, fixing $ n $ and $ \ell $ (thus the information rate), we may only focus on communication overheads, since both behave equally.

To measure the quality of a coset coding scheme, we need the following two bounds. The first is given in \cite[Th. 12]{silva-universal} and can be seen as a $ q $-analog of the bound in \cite[Prop. 1]{efficient}, although considering also errors and not only erasures:

\begin{proposition}[\textbf{\cite{silva-universal}}] \label{proposition bound on info rate}
If the coset coding scheme $ F : \mathbb{F}_q^{\alpha m \times \ell} \longrightarrow \mathbb{F}_q^{\alpha m \times n} $ is universally $ t $-error and $ \rho $-erasure-correcting, and universally secure under $ \mu $ observations, then
\begin{equation}
\ell \leq n- 2t - \rho - \mu. \label{eq lower bound info rate}
\end{equation}
\end{proposition}

Next we give a $ q $-analog of the bound in \cite[Th. 1]{efficient}, again adding the effect of errors, which was not considered in \cite{efficient}:

\begin{proposition} \label{lemma lower bounds}
If the coset coding scheme $ F : \mathbb{F}_q^{\alpha m \times \ell} \longrightarrow \mathbb{F}_q^{\alpha m \times n} $ is universally secure under $ \mu $ observations, then for a full-rank matrix $ A \in \mathbb{F}_q^{d \times n} $ and preprocessing functions $ E_{A,i} : \mathbb{F}_q^{\alpha m} \longrightarrow \mathbb{F}_q^{\beta_i m} $, for $ i = 1,2, \ldots, d $, that are $ t $-error-correcting with respect to $ F $, it holds that:
\begin{equation}
{\rm CO}(A) \geq \frac{\ell (2t + \mu)}{d - 2t - \mu}, \label{eq lower bound CO}
\end{equation}
\end{proposition}
\begin{proof}
We may assume without loss of generality that $ \beta_1 \leq \beta_2 \leq \ldots \leq \beta_d $ as in the proof of \cite[Th. 1]{efficient}. 

First, we prove that the preprocessing functions $ E_{A,i} : \mathbb{F}_q^{\alpha m} \longrightarrow \mathbb{F}_q^{\beta_i m} $, for $ i = 1,2, \ldots, d-2t $ are $ 0 $-error-correcting with respect to $ F : \mathbb{F}_q^{\alpha m \times \ell} \longrightarrow \mathbb{F}_q^{\alpha m \times n} $. If they were not, then there would exist $ C_1 \in \mathcal{C}_{S_1} $ and $ C_2 \in \mathcal{C}_{S_2} $, with $ S_1 \neq S_2 $, such that
$$ \left( E_{A,i} \left( C_1 \mathbf{a}_i^T \right) \right)_{i=1}^{d-2t} = \left( E_{A,i} \left( C_2 \mathbf{a}_i^T \right) \right)_{i=1}^{d-2t}. $$
On the other hand, there exist $ \mathbf{e}_i \in \mathbb{F}_q^{\alpha m} $ such that $ C_1 \mathbf{a}_i^T + \mathbf{e}_i = C_2 \mathbf{a}_i^T $, for $ i = d-2t+1, d-2t+2, \ldots, d-t $, and $ C_1 \mathbf{a}_i^T = C_2 \mathbf{a}_i^T + \mathbf{e}_i $, for $ i = d-t+1, d-t+2, \ldots, d $. Thus we see that the preprocessing functions $ E_{A,i} : \mathbb{F}_q^{\alpha m} \longrightarrow \mathbb{F}_q^{\beta_i m} $, for $ i = 1,2, \ldots, d $, cannot be $ t $-error-correcting with respect to $ F $, which is a contradiction.

Next, defining $ f_i = E_{A,i} \left( C \mathbf{a}_i^T \right) $, where $ C = F(S) $, for $ i = 1,2, \ldots, d-2t $ and $ S \in \mathbb{F}_q^{\alpha m \times \ell} $, we may prove exactly as in the proof of \cite[Th. 1]{efficient} that
\begin{equation}
 \sum_{i = 1}^{d-2t} \frac{\beta_i}{\alpha} \geq \frac{\ell (d - 2t)}{d - 2t - \mu},
\label{eq proof of lower bound 1}
\end{equation}
and also
\begin{equation}
 \frac{\beta_{d-2t - \mu}}{\alpha} \geq \frac{\ell}{d - 2t - \mu}.
\label{eq proof of lower bound 2}
\end{equation}
Now using that $ \beta_{d-2t-\mu} \leq \beta_{d-2t-\mu+1} \leq \ldots \leq \beta_d $, and combining Equations (\ref{eq proof of lower bound 1}) and (\ref{eq proof of lower bound 2}), we conclude that
$$ {\rm DB}(A) = \left( \sum_{i = 1}^{d-2t} \frac{\beta_i}{\alpha} \right) + \left( \sum_{i = d-2t +1}^{d} \frac{\beta_i}{\alpha} \right) $$
$$ \geq \frac{\ell (d - 2t)}{d - 2t - \mu} + 2t \frac{\ell}{d - 2t - \mu} = \frac{\ell d}{d - 2t - \mu}, $$
and the bound on $ {\rm CO}(A) $ follows by substracting $ \ell $ to this inequality.
\end{proof}

\section{A general construction based on linear rank-metric codes} \label{sec general construction}

In this section, given a nested coset coding scheme (Definition \ref{def nested schemes}) able to correct $ t_0 $ errors and $ \rho_0 $ erasures, for fixed positive integers $ t_0 $ and $ \rho_0 $, and given an arbitrary set $ \mathcal{D} \subseteq [n-\rho_0, n] $ such that $ n-\rho_0 \in \mathcal{D} $ (in particular for $ \mathcal{D} = [n-\rho_0, n] $), we construct a coset coding scheme able to correct $ t_0 $ errors and any $ n-d $ erasures with lower communication overheads than the original scheme, for all $ d \in \mathcal{D} $. Moreover, both the original scheme and the modified one are universally secure under the same number of observations. The downside of the method is multiplying the packet length of the original coset coding scheme by a parameter $ \alpha $, depending on the involved codes, to achieve the desired subpacketization. The main result of the section is the following:

\begin{theorem} \label{theorem general univeral staircase properties}
Take $ \mathbb{F}_{q^m} $-linear codes $ \mathcal{C}_2 \subsetneqq \mathcal{C}_1 \subseteq \mathbb{F}_{q^m}^n $, a positive integer $ \rho_0 $ such that $ \rho_0 < d_R(\mathcal{C}_1, \mathcal{C}_2) $, and choose any subset $ \mathcal{D} \subseteq [n-\rho_0, n] $ such that $ n-\rho_0 \in \mathcal{D} $. Denote $ \mathcal{D} = \{ d_1, d_2, \ldots , d_h \} $, where $ d_h = n-\rho_0 < d_{h-1} < \ldots < d_2 < d_1 $, and assume that there exists a sequence of nested $ \mathbb{F}_{q^m} $-linear codes
$$ \mathcal{C}_1 = \mathcal{C}^{(h)} \subsetneqq \mathcal{C}^{(h-1)} \subsetneqq \ldots \subsetneqq \mathcal{C}^{(2)} \subsetneqq \mathcal{C}^{(1)} \subseteq \mathbb{F}_{q^m}^n $$
such that 
$$ d_R \left( \mathcal{C}^{(j)}, \mathcal{C}_2 \right) \geq n - d_j + 1, $$
for $ j = 1,2, \ldots, h $. Define then $ k_1 = \dim(\mathcal{C}_1) $, $ k_2 = \dim(\mathcal{C}_2) $, $ \ell = k_1 - k_2 $, $ \alpha_j = k^{(j)} - k_2 $ for $ j = 1,2, \ldots, h $, and 
$$ \alpha = {\rm LCM}(\alpha_1, \alpha_2, \ldots, \alpha_h). $$

There exists a coset coding scheme $ F : \mathbb{F}_{q^m}^{\alpha \times \ell} \longrightarrow \mathbb{F}_{q^m}^{\alpha \times n} $ that is universally $ t $-error and $ \rho $-erasure-correcting if $ 2t + \rho < d_R \left( \mathcal{C}_1, \mathcal{C}_2 \right) $, and is universally secure under $ \mu $ observations if $ \mu < d_R \left( \mathcal{C}_2^\perp, \mathcal{C}^{(1) \perp} \right) $.

In addition, for any $ d \in \mathcal{D} $ and any full-rank matrix $ A \in \mathbb{F}_q^{d \times n} $, there exist preprocessing functions $ E_{A,i} : \mathbb{F}_q^{\alpha m} \longrightarrow \mathbb{F}_q^{\beta_i m} $, for $ i = 1,2, \ldots, d $, which are $ t $-error-correcting with respect to $ F $, whenever $ 2t < d_R \left( \mathcal{C}^{(j)}, \mathcal{C}_2 \right) - n + d $, and such that
$$ {\rm CO}(A) = \frac{\ell \left( d - k^{(j)} + k_2 \right)}{k^{(j)} - k_2}, $$
where $ k^{(j)} = \dim \left( \mathcal{C}^{(j)} \right) $, for $ j $ such that $ d = d_j $.
\end{theorem}

\subsection{Description of the construction for Theorem \ref{theorem general univeral staircase properties}} \label{subsec description general}

Let the notation be as in Theorem \ref{theorem general univeral staircase properties} and take a generator matrix $ G_2 \in \mathbb{F}_{q^m}^{k_2 \times n} $ of $ \mathcal{C}_2 $ and a generator matrix $ G_1 \in \mathbb{F}_{q^m}^{k_1 \times n} $ of $ \mathcal{C}_1 $ of the form
\begin{displaymath}
G_1 = \left(
\begin{array}{c}
G_2 \\
G_c
\end{array} \right) \in \mathbb{F}_{q^m}^{k_1 \times n},
\end{displaymath}
for some matrix $ G_c \in \mathbb{F}_{q^m}^{\ell \times n} $. Decreasingly in $ j = h-1,h-2, \ldots, 2,1 $, take a generator matrix $ G^{(j)} \in \mathbb{F}_{q^m}^{k^{(j)} \times n} $ of $ \mathcal{C}^{(j)} $ of the form
\begin{displaymath}
G^{(j)} = \left(
\begin{array}{c}
G^{(j+1)} \\
G^{(j+1)}_c
\end{array} \right) \in \mathbb{F}_{q^m}^{k^{(j)} \times n},
\end{displaymath} 
for some matrix $ G_c^{(j+1)} \in \mathbb{F}_{q^m}^{\left( k^{(j)} - k^{(j+1)} \right) \times n} $. Next define the following positive integers, which are analogous to the integers defined in \cite[Eq. (11)]{efficient}:
\begin{displaymath}
p_j = \left\lbrace \begin{array}{ll}
\frac{\ell \alpha}{\alpha_1} & \textrm{if } j = 1, \\
\frac{\ell \alpha}{\alpha_j} - \frac{\ell \alpha}{\alpha_{j-1}} & \textrm{if } 1 < j \leq h.
\end{array} \right.
\end{displaymath}

Let $ S \in \mathbb{F}_{q^m}^{\alpha \times \ell} $ be the secret message and generate uniformly at random a matrix $ R \in \mathbb{F}_{q^m}^{\alpha \times k_2} $. Divide $ S $ and $ R $ as follows:
\begin{displaymath}
S = \left(
\begin{array}{c}
S_1 \\
S_2 \\
\vdots \\
S_h 
\end{array} \right), \quad R = \left(
\begin{array}{c}
R_1 \\
R_2 \\
\vdots \\
R_h 
\end{array} \right),
\end{displaymath}
where
$$ S_j \in \mathbb{F}_{q^m}^{p_j \times \ell} \textrm{ and } R_j \in \mathbb{F}_{q^m}^{p_j \times k_2}, $$
for $ j = 1,2, \ldots, h $. Next, we define the matrices
\begin{displaymath}
\begin{array}{cccccccc}
M_1 & = & (R_1 & S_1 & D_{1,1} & D_{1,2} & \ldots & D_{1,h-1} ), \\
M_2 & = & (R_2 & S_2 & D_{2,1} & D_{2,2} & \ldots & 0 ), \\
M_3 & = & (R_3 & S_3 & D_{3,1} & D_{3,2} & \ldots & 0 ), \\
\vdots & & & \vdots & & & & \vdots \\
M_{h-1} & = & (R_{h-1} & S_{h-1} & D_{h-1,1} & 0 & \ldots & 0 ), \\
M_h & = & (R_h & S_h & 0 & 0 & \ldots & 0 ), \\
\end{array} 
\end{displaymath}
where $ M_u \in \mathbb{F}_{q^m}^{p_u \times k^{(1)}} $, and where the matrices $ D_{u, v} \in \mathbb{F}_{q^m}^{p_u \times (\alpha_{h-v} - \alpha_{h-v+1})} $ are defined iteratively as follows: For $ v = 1,2, \ldots, h-1 $, the components of the $ v $-th column block
\begin{displaymath}
\left( \begin{array}{c}
D_{1, v} \\
D_{2, v} \\
\vdots \\
D_{h-v, v}
\end{array} \right) \in \mathbb{F}_{q^m}^{\ell \alpha / \alpha_{h-v} \times (\alpha_{h-v} - \alpha_{h-v+1})},
\end{displaymath}
are the components (after some fixed rearrangement) of 
$$ (S_{h-v+1} | D_{h-v+1, 1} | D_{h-v+1, 2} | \ldots | D_{h-v+1, v-1}), $$
whose size is $ p_{h-v+1} \times \alpha_{h-v+1} $ (observe that $ p_{j+1} \alpha_{j+1} = (\alpha_j - \alpha_{j+1}) \ell \alpha / \alpha_j $). For convenience, we define the matrices
\begin{equation}
M^\prime_j = \left(
\begin{array}{c}
M_1 \\
M_2 \\
\vdots \\
M_j  
\end{array} \right) \in \mathbb{F}_{q^m}^{\ell \alpha / \alpha_j \times k^{(1)}}, \label{eq definition M_h}
\end{equation}
for $ j = 1,2, \ldots, h $. 

Finally, we define the coset coding scheme $ F : \mathbb{F}_{q^m}^{\alpha \times \ell} \longrightarrow \mathbb{F}_{q^m}^{\alpha \times n} $ by
\begin{equation}
C = F(S) = M^\prime_h G^{(1)} \in \mathbb{F}_{q^m}^{\alpha \times n}.
\label{eq construction definition}
\end{equation}

To conclude, we define $ E_{A,i} : \mathbb{F}_{q^m}^\alpha \longrightarrow \mathbb{F}_{q^m}^{\ell \alpha / \alpha_j} $ as follows. For $ j = 1,2, \ldots, h $, for $ i = 1,2, \ldots, d_j $, and for a full-rank matrix $ A \in \mathbb{F}_q^{d_j \times n} $, we define $ E_{A,i}(\mathbf{y}_i) \in \mathbb{F}_{q^m}^{\ell \alpha / \alpha_j} $ by restricting $ \mathbf{y}_i \in \mathbb{F}_{q^m}^\alpha $ to its first $ \ell \alpha / \alpha_j $ rows.

\subsection{Proof of Theorem \ref{theorem general univeral staircase properties}} \label{subsec proof general}

Let the notation be as in Theorem \ref{theorem general univeral staircase properties} and as in the previous subsection. We prove each statement in Theorem \ref{theorem general univeral staircase properties} separately:

\textit{1) The coset coding scheme is universally $ t $-error and $ \rho $-erasure-correcting if $ 2t + \rho < d_R(\mathcal{C}_1, \mathcal{C}_2) $:} Take $ A \in \mathbb{F}_q^{N \times n} $ of rank at least $ n-\rho $ and an error matrix $ E \in \mathbb{F}_{q^m}^{\alpha \times N} $ such that $ {\rm Rk}_q(E) \leq t $. Divide $ E $ in the same way as $ S $ and $ R $, that is,
\begin{displaymath}
E = \left(
\begin{array}{c}
E_1 \\
E_2 \\
\vdots \\
E_h 
\end{array} \right) \in \mathbb{F}_{q^m}^{\alpha \times N},
\end{displaymath}
where $ E_j \in \mathbb{F}_{q^m}^{p_j \times N} $, and observe that $ {\rm Rk}_q(E_j) \leq {\rm Rk}_q(E) \leq t $, for $ j = 1,2, \ldots, h $. From 
$$ (S_h | R_h | 0) G^{(1)} A^T + E_h = (S_h | R_h) G_1 A^T + E_h $$ 
we obtain $ S_h $ by Lemma \ref{lemma correction capability}, since $ {\rm Rk}_q(E_h) \leq t $ and $ 2t+\rho < d_R(\mathcal{C}_1, \mathcal{C}_2) $.  By definition, we have obtained $ D_{u,1} $, for $ u = 1,2, \ldots h-1 $. Hence substracting $ D_{h-1,1} G_c^{(h)} A^T $ from $ (S_{h-1} | R_{h-1} | D_{h-1,1} | $ $ 0)G^{(1)} A^T + E_{h-1} $, we may obtain $ (S_{h-1} | R_{h-1} ) G_1 A^T + E_{h-1} $, and thus we obtain $ S_{h-1} $ again by Lemma \ref{lemma correction capability}. Now, we have also obtained $ D_{u,2} $, for $ u = 1,2, \ldots, h-2 $. Proceeding iteratively in the same way, we see that we may obtain all the matrices $ S_j $, for $ j = 1,2, \ldots, h $, and thus we obtain the whole message $ S $.

\textit{2) The coset coding scheme is universally secure under any $ \mu < d_R \left( \mathcal{C}_2^\perp, \mathcal{C}_1^\perp \right) $ observations:} We first need the following preliminary lemma, which follows from \cite[Th. 3]{similarities}: % follows easily viewing the minimum rank distance of a code as its first generalized rank weight (see \cite[Def. 5 and Cor. 13]{rgrw}):

\begin{lemma} \label{lemma secrecy}
Let $ B \in \mathbb{F}_q^{\mu \times n} $ and let $ \mathcal{C}_2 \subsetneqq \mathcal{C}_1 \subseteq \mathbb{F}_{q^m}^n $ be $ \mathbb{F}_{q^m} $-linear codes. If $ {\rm Rk}(B) < d_R \left( \mathcal{C}_2^\perp, \mathcal{C}_1^\perp \right) $, then
$$ \mathcal{C}_2 B^T = \mathcal{C}_1 B^T, $$
where $ \mathcal{C} B^T = \left\lbrace \mathbf{c} B^T \mid \mathbf{c} \in \mathcal{C} \right\rbrace \subseteq \mathbb{F}_{q^m}^\mu $, for a code $ \mathcal{C} \subseteq \mathbb{F}_{q^m}^n $.
\end{lemma}
\begin{proof}
See Appendix \ref{app proof of}.
\end{proof}

Take $ B \in \mathbb{F}_q^{\mu \times n} $, and assume that the eavesdropper obtains
$$ W = CB^T = M^\prime_h G^{(1)} B^T \in \mathbb{F}_{q^m}^{\alpha \times \mu}. $$
The random variable $ W $ has support inside the $ \mathbb{F}_{q^m} $-linear vector space
$$ \mathcal{C}^{(1)}_{B, \alpha} = \left\lbrace M G^{(1)} B^T \mid M \in \mathbb{F}_{q^m}^{\alpha \times k^{(1)}} \right\rbrace \subseteq \mathbb{F}_{q^m}^{\alpha \times \mu}. $$
Recall from \cite[Th. 2.6.4]{cover} that, if a random variable $ X $ has support in the set $ \mathcal{X} $, then $ H(X) \leq \log_q(\# \mathcal{X}) $. Hence
$$ H(W) \leq \log_{q} \left( \# \mathcal{C}^{(1)}_{B, \alpha} \right) = m \dim \left( \mathcal{C}^{(1)}_{B, \alpha} \right) = \alpha m \dim \left( \mathcal{C}^{(1)}B^T \right), $$
where dimensions are taken over $ \mathbb{F}_{q^m} $. On the other hand, using the analogous notation $ \mathcal{C}_{2 B, \alpha} $ for $ \mathcal{C}_2 $ instead of $ \mathcal{C}^{(1)} $, it holds that
$$ H(W \mid S) = \log_q \left( \# \mathcal{C}_{2 B, \alpha} \right) = m \dim \left( \mathcal{C}_{2 B, \alpha} \right) = \alpha m \dim \left( \mathcal{C}_2 B^T \right), $$
since, given a value of $ S $, the variable $ W $ is a uniform random variable over an $ \mathbb{F}_{q^m} $-linear affine space obtained by translating the vector space $ \mathcal{C}_{2 B, \alpha} $. Hence we obtain that 
$$ 0 \leq I(S; W) = H(W) - H(W \mid S) $$
$$ \leq \alpha m \left( \dim \left( \mathcal{C}^{(1)} B^T \right) - \dim \left( \mathcal{C}_2 B^T \right) \right) = 0, $$
where the last equality follows from Lemma \ref{lemma secrecy}. Thus $ I(S; W) = 0 $ and we are done.

\textit{3) The preprocessing functions are $ t $-error-correcting for any $ 2t < d_R \left( \mathcal{C}^{(j)}, \mathcal{C}_2 \right) - n + d $, where $ d = d_j $:} Fix $ d \in \mathcal{D} $ and a full-rank matrix $ A \in \mathbb{F}_q^{d \times n} $, and let $ E_{A,i} : \mathbb{F}_{q^m}^\alpha \longrightarrow \mathbb{F}_{q^m}^{\ell \alpha / \alpha_j} $ be preprocessing functions as in the previous subsection, for $ i = 1,2, \ldots, d $, and where $ j $ is such that $ d = d_j $. 

Let $ E \in \mathbb{F}_{q^m}^{\alpha \times d} $ be an error matrix such that $ {\rm Rk}_q(E) \leq t $, and let $ \mathbf{e}_1, \mathbf{e}_2, \ldots, \mathbf{e}_d \in \mathbb{F}_{q^m}^\alpha $ be its columns. By definition, $ E_{A,i} \left( C \mathbf{a}_i^T + \mathbf{e}_i \right) $ is the $ i $-th column of
$$ M_j^\prime G^{(1)} A^T + E_j^\prime \in \mathbb{F}_{q^m}^{\ell \alpha / \alpha_j \times d}, $$
for $ i = 1,2, \ldots, d $, and a submatrix $ E^\prime_j \in \mathbb{F}_{q^m}^{\ell \alpha / \alpha_j \times d} $ of $ E $, which thus satisfies that $ {\rm Rk}_q(E_j^\prime) \leq {\rm Rk}_q(E) \leq t $. Therefore, we may obtain the matrix $ M_j^\prime $ as in item 1, since $ 2t + n - d < d_R(\mathcal{C}^{(j)}, \mathcal{C}_2) $. By definition, the matrices $ S_1, S_2, \ldots, S_j $ are contained in $ M_j^\prime $. Moreover, the matrices $ D_{1,h-j}, D_{2,h-j}, \ldots, D_{j,h-j} $ are also contained in $ M_j^\prime $, and from them we obtain by definition $ S_{j+1} $ and $ D_{j+1,1}, D_{j+1,2}, \ldots, D_{j+1,h-j-1} $. Now, the matrices $ D_{1,h-j-1}, D_{2,h-j-1}, \ldots, D_{j,h-j-1} $ are contained in $ M_j^\prime $ and we also have $ D_{j+1,h-j-1} $, hence we may obtain by definition $ S_{j+2} $ and $ D_{j+2,1}, D_{j+2,2}, \ldots, D_{j+2,h-j-2} $. Continuing iteratively in this way, we may obtain all $ S_1, S_2, \ldots, S_h $ and hence the message $ S $. 

Finally, we have that
$$ {\rm CO}(A) = \sum_{i = 1}^d \frac{\beta_i}{\alpha} - \ell = \sum_{i = 1}^d \frac{\ell \alpha}{\alpha_j \alpha} - \ell $$
$$ = \frac{\ell d}{\alpha_j} - \ell = \frac{\ell (d - \alpha_j)}{\alpha_j} = \frac{\ell \left( d - k^{(j)} + k_2 \right)}{k^{(j)} - k_2}. $$

\section{MRD codes and coset coding schemes with optimal communication overheads} \label{sec optimal schemes}

In this section, we apply Theorem \ref{theorem general univeral staircase properties} to pairs of Gabidulin codes \cite{gabidulin, roth} and their cartesian products \cite{reducible}. The first family yields optimal coset coding schemes when $ n \leq m $ in the sense of (\ref{eq lower bound info rate}) and (\ref{eq lower bound CO}), and the second family constitutes a family of maximum rank distance (MRD) codes when $ n > m $ \cite[Cor. 1]{reducible}.%, hence constitutes the best option in this case among $ \mathbb{F}_{q^m} $-linear codes due to Lemma \ref{lemma correction capability}.

We recall the definition of MRD codes for convenience of the reader. The Singleton bound for an arbitrary (linear or not) code $ \mathcal{C} \subseteq \mathbb{F}_{q^m}^n $ was first given in \cite[Th. 6.3]{delsartebilinear}:
\begin{equation}
\# \mathcal{C} \leq q^{\max\{ m,n \} \left( \min \{ m,n \} - d_R(\mathcal{C}) + 1 \right)}.
\label{eq singleton bound general}
\end{equation}
We then say that $ \mathcal{C} $ is MRD if equality holds in (\ref{eq singleton bound general}). In another direction, a Singleton bound on the relative minimum rank distance of a pair of $ \mathbb{F}_{q^m} $-linear codes $ \mathcal{C}_2 \subsetneqq \mathcal{C}_1 \subseteq \mathbb{F}_{q^m}^n $ was first given in \cite[Prop. 3]{rgrw}:
\begin{equation}
d_R(\mathcal{C}_1, \mathcal{C}_2) \leq \min \left\lbrace n - \dim(\mathcal{C}_1), \frac{m (n - \dim(\mathcal{C}_1))}{n - \dim(\mathcal{C}_2)} \right\rbrace + 1.
\label{eq singleton bound relative}
\end{equation}
Thus if $ \mathcal{C}_1 $ is MRD and $ n \leq m $, then equality is satisfied in (\ref{eq singleton bound relative}).

\subsection{Coset coding schemes based on Gabidulin codes} \label{subsec using gabidulin}

In this subsection we will make use of Gabidulin codes, which were introduced independently in \cite[Sec. 4]{gabidulin} and \cite[Sec. III]{roth}. Throughout this subsection, we will assume that $ n \leq m $.

\begin{definition} [\textbf{\cite{gabidulin, roth}}] \label{def gabidulin code}
Fix a basis $ \gamma_1, \gamma_2, \ldots, \gamma_m $ of $ \mathbb{F}_{q^m} $ as a vector space over $ \mathbb{F}_q $, and let $ 0 \leq k \leq n $. The Gabidulin code of dimension $ k $ and length $ n $ over $ \mathbb{F}_{q^m} $, constructed from the previous basis, is the $ \mathbb{F}_{q^m} $-linear code $ \mathcal{G}_k \subseteq \mathbb{F}_{q^m}^n $ with parity-check matrix given by
\begin{displaymath}
\left( \begin{array}{ccccc}
\gamma_1 & \gamma_2 & \gamma_3 & \ldots & \gamma_n \\
\gamma_1^q & \gamma_2^q & \gamma_3^q & \ldots & \gamma_n^q \\
\gamma_1^{q^2} & \gamma_2^{q^2} & \gamma_3^{q^2} & \ldots & \gamma_n^{q^2} \\
\vdots & \vdots & \vdots & \ddots & \vdots \\
\gamma_1^{q^{n-k-1}} & \gamma_2^{q^{n-k-1}} & \gamma_3^{q^{n-k-1}} & \ldots & \gamma_n^{q^{n-k-1}} \\
\end{array} \right) \in \mathbb{F}_{q^m}^{(n-k) \times n}.
\end{displaymath}
\end{definition}
 
It was proven in \cite[Th. 6]{gabidulin} and \cite[Th. 2]{roth} that the code $ \mathcal{G}_k \subseteq \mathbb{F}_{q^m}^n $ satisfies
\begin{equation}
\dim \left( \mathcal{G}_k \right) = k, \quad \textrm{and} \quad d_R(\mathcal{G}_k) = n-k+1,
\label{eq properties gabidulin codes}
\end{equation}
constituting thus a family of MRD codes covering all parameters when $ n \leq m $. Moreover it is clear from the definition that, for a fixed basis of $ \mathbb{F}_{q^m} $ over $ \mathbb{F}_q $, they form a nested sequence of codes:
\begin{equation}
\{ \mathbf{0} \} = \mathcal{G}_0 \subsetneqq \mathcal{G}_1 \subsetneqq \mathcal{G}_2 \subsetneqq \ldots \subsetneqq \mathcal{G}_{n-1} \subsetneqq \mathcal{G}_n = \mathbb{F}_{q^m}^n.
\label{eq sequence of gabidulin}
\end{equation}
Thus the next theorem follows directly from Theorem \ref{theorem general univeral staircase properties}:

\begin{theorem} \label{theorem using gabidulin}
Choose integers $ k_2, k_1, t_0 $ and $ \rho_0 $ such that $ 0 \leq k_2 < k_1 \leq n $ and $ 2 t_0 + \rho_0 = n - k_1 $, and choose any subset $ \mathcal{D} \subseteq [n-\rho_0, n] $ such that $ n-\rho_0 \in \mathcal{D} $. 

Now, fix a basis of $ \mathbb{F}_{q^m} $ over $ \mathbb{F}_q $, let $ \mathcal{C}_2 \subsetneqq \mathcal{C}_1 \subseteq \mathbb{F}_{q^m}^n $ be $ \mathbb{F}_{q^m} $-linear Gabidulin codes of dimensions $ k_2 $ and $ k_1 $ (that is, $ \mathcal{G}_{k_2} $ and $ \mathcal{G}_{k_1} $), respectively, and denote the elements in $ \mathcal{D} $ by $ d_h = n-\rho_0 < d_{h-1} < \ldots < d_2 < d_1 $.

The coset coding scheme $ F : \mathbb{F}_{q^m}^{\alpha \times \ell} \longrightarrow \mathbb{F}_{q^m}^{\alpha \times n} $ in Theorem \ref{theorem general univeral staircase properties} based on this pair of codes and the subsequence of (\ref{eq sequence of gabidulin}) given by the Gabidulin codes $ \mathcal{C}^{(j)} = \mathcal{G}_{d_j - 2t_0} $, that is, $ k^{(j)} = d_j - 2 t_0 $, for $ j = 1,2, \ldots, h $, satisfies $ \ell = k_1-k_2 $, is universally $ t $-error and $ \rho $-erasure-correcting if $ 2t + \rho \leq n-k_1 $, and is universally secure under $ \mu $ observations if $ \mu \leq k_2 $. In particular, the scheme is optimal in the sense of (\ref{eq lower bound info rate}). Moreover, it holds that
$$ \alpha = {\rm LCM} \left( d_1 - 2t_0 - k_2, d_2 - 2t_0 - k_2, \ldots, d_h - 2t_0 - k_2 \right). $$

In addition, for any $ d \in \mathcal{D} $ and any full-rank matrix $ A \in \mathbb{F}_q^{d \times n} $, there exist preprocessing functions $ E_{A,i} : \mathbb{F}_q^{\alpha m} \longrightarrow \mathbb{F}_q^{ \ell \alpha m / (d - 2 t_0 - k_2)} $, for $ i = 1,2, \ldots, d $, which are $ t_0 $-error-correcting and satisfying equality in (\ref{eq lower bound CO}), hence having optimal communication overheads for all $ d \in \mathcal{D} $.
\end{theorem}

Observe that the packet length $ m $ of the original Gabidulin codes is multiplied by $ \alpha $, which depends only on the maximum number of observations, the number of correctable errors and the set of possible erasures $ \mathcal{D} $.

However, there are instances as Example \ref{example 1} where, due to a particular subpacketization, we need not expand the packet length, hence we obtain a strict improvement on the communication overheads at no cost on the rest of the parameters.

We now give the details of Example \ref{example 1} and Example \ref{example 2}, which share the same construction: With the given parameters, the construction in \cite[Th. 11]{silva-universal} gives $ \ell = 16 $ by choosing $ k_1 = 24 $ and $ k_2 = 8 $. However, decomposing the packet length as $ \alpha m = 2048 $, with $ m = 64 $ and $ \alpha = 32 $, we may choose $ \mathcal{D} = \{ 24, 40 \} $, $ k^{(1)} = 40 $, $ k_1 = 24 $ and $ k_2 = 8 $, thus $ \alpha_1 = 32 $, $ \alpha_2 = 16 $, and $ \alpha = 32 $, and the example follows.

\subsection{Coset coding schemes based on MRD cartesian products}

In this subsection, we will make use of cartesian products of Gabidulin codes, which yield again MRD codes, but in the case $ n > m $, in contrast with plain Gabidulin codes as in the previous subsection. To the best of our knowledge, this is the only known family of MRD $ \mathbb{F}_{q^m} $-linear codes in $ \mathbb{F}_{q^m}^n $ when $ n > m $.

Throughout this subsection, we will assume that $ n = lm $, for some positive integer $ l $. Take another integer $ 1 \leq k \leq m $, and consider the cartesian product
$$ \mathcal{C} = \mathcal{G}_{k}^l \subseteq \mathbb{F}_{q^m}^n, $$
where $ \mathcal{G}_{k} \subseteq \mathbb{F}_{q^m}^m $ is a Gabidulin code as in Definition \ref{def gabidulin code}. It is proven in \cite[Cor. 1]{reducible} that
\begin{equation}
\dim(\mathcal{C}) = l k, \quad \textrm{and} \quad d_R(\mathcal{C}) = m - k + 1,
\label{eq properties reducible codes}
\end{equation}
and therefore $ \mathcal{C} $ is MRD. Since the codes $ \mathcal{G}_{k} $ can be taken in a nested sequence for a fixed basis of $ \mathbb{F}_{q^m} $ over $ \mathbb{F}_q $, as in Equation (\ref{eq sequence of gabidulin}), the next result also follows directly from Theorem \ref{theorem general univeral staircase properties}:

\begin{theorem}
Choose integers $ k_1, k_2, t_0 $ and $ \rho_0 $ such that $ 0 \leq k_2 < k_1 \leq m $ and $ 2t_0 + \rho_0 = m - k_1 $, and choose any subset $ \mathcal{D} \subseteq [n-\rho_0,n] $ with elements $ d_h = n-\rho_0 < d_{h-1} < \ldots < d_2 < d_1 $.

Define $ k^{(j)} = d_j - (l-1) m - 2t_0 $ and the $ \mathbb{F}_{q^m} $-linear codes
$$ \mathcal{C}_2 = \mathcal{G}_{k_2}^l \subsetneqq \mathcal{C}^{(j)} = \mathcal{G}_{k^{(j)}}^l \subseteq \mathbb{F}_{q^m}^n, $$
for $ j = 1,2, \ldots, h $, and observe that $ k^{(h)} = k_1 $, hence $ \mathcal{C}^{(h)} = \mathcal{C}_1 = \mathcal{G}_{k_1}^l $. 

The coset coding scheme $ F : \mathbb{F}_{q^m}^{\alpha \times \ell} \longrightarrow \mathbb{F}_{q^m}^{\alpha \times n} $ in Theorem \ref{theorem general univeral staircase properties} based on these codes satisfies $ \ell = l (k_1-k_2) $, is universally $ t $-error and $ \rho $-erasure-correcting if $ 2t + \rho \leq m-k_1 $, and is universally secure under $ \mu $ observations if $ \mu \leq k_2 $. Moreover, it holds that
$$ \alpha = {\rm LCM} \left\lbrace l ( d_j  - 2t_0 - k_2 ) - (l-1) n \mid j =1,2, \ldots, h \right\rbrace. $$

In addition, for any $ d \in \mathcal{D} $ and any full-rank matrix $ A \in \mathbb{F}_q^{d \times n} $, there exist preprocessing functions $ E_{A,i} : \mathbb{F}_q^{\alpha m} \longrightarrow \mathbb{F}_q^{ \ell \alpha m / (l(d - 2 t_0 - k_2) - (l-1)n)} $, for $ i = 1,2, \ldots, d $, which are $ t_0 $-error-correcting and such that
$$ {\rm CO}(A) = \frac{\ell \left( l (2t_0 + k_2) + (l-1)(n-d) \right)}{l (d - 2t_0 - k_2) - (l-1)n}. $$
\end{theorem}

Observe that the particular case $ l = 1 $ corresponds to the particular case $ n=m $ in Theorem \ref{theorem using gabidulin}.

\section{Applications} \label{sec applications}

\subsection{Universal secure linear network coding}

Consider a network with $ n $ outgoing links from a source and $ N $ ingoing links to a sink, and where the source wants to transmit $ \ell $ packets, encoded into $ n $ packets (all of the same length), to the sink. Linear network coding, introduced in \cite{ahlswede, Koetter2003, linearnetwork}, consists in sending linear combinations over $ \mathbb{F}_q $ of the received packets at each node of the network, which increases throughput with respect to storing and forwarding.

In this scenario, link errors and erasures expand through the network and an eavesdropper may obtain linear combinations of the sent packets. Thus if the coefficients of the final linear combinations are known to the receiver, then a linearly coded network, with link errors, erasures and observations, can be modelled as a coherent linearized noisy wire-tap channel \cite{silva-universal}, as in Definition \ref{def lin channel}. 

Assume that the packet length is at least $ n $, and fix positive integers $ t $, $ \rho $ and $ \mu $ with $ 2t + \rho + \mu < n $ and $ \rho \leq N $. In \cite[Th. 11]{silva-universal} a construction (pairs of Gabidulin codes) is given such that $ \ell = n - 2t - \rho - \mu $, which is optimal due to (\ref{eq lower bound info rate}). 

However, assuming that $ q $ is big enough and the erasure matrix $ A \in \mathbb{F}_q^{N \times n} $ (see Definition \ref{def lin channel}) is taken at random as in \cite{random}, then it will be full-rank with high probability and $ \rho $ can be thought of as a number of erased ingoing links to the sink, due to noise, link failure or the action of the adversary. 

Theorem \ref{theorem using gabidulin} gives an alternative construction to \cite[Th. 11]{silva-universal} with optimal $ \ell = n - 2t - \rho - \mu $, where if more than $ n - \rho $ ingoing links to the sink are available, the sink can contact the corresponding nodes after exchanging feedback on the number of available nodes, and reduce the communication overhead (hence the amount of packets received by the sink) to its optimal value in view of (\ref{eq lower bound CO}).

\subsection{Secure distributed storage with crisscross errors and erasures} \label{subsec distributed}

Errors and erasures occurring along several rows and/or columns of a matrix over $ \mathbb{F}_q $ are called \textit{crisscross errors and erasures} in the literature, and can happen in memory chips and magnetic tapes, for instance (see \cite{roth}). Recently, crisscross error and erasure-correction has gained attention in the context of distributed storage where data is stored in several data centers (columns), which in turn store several blocks of data (rows), where mixed and/or correlated failures may occur (see \cite{swanan}).

In this work, we consider a storage system where data is stored as an $ \alpha m \times n $ matrix over $ \mathbb{F}_q $, where columns are thought of as data centers that are contacted to obtain information from, and rows are blocks of data expanding across the different data centers and sharing correlated errors. More formally, we consider column erasures (equivalently, data centers being available and contacted) together with crisscross errors and where an eavesdropper may listen to a number of columns (data centers). 

We formalize crisscross error-correction in the following definitions, which we take from \cite[Sec. I]{roth}:

\begin{definition} [\textbf{Crisscross weights \cite{roth}}]
A cover of a matrix $ E \in \mathbb{F}_q^{\alpha m \times n} $ is a pair of sets $ X \subseteq [\alpha m] $ and $ Y \subseteq [n] $ such that if $ e_{i,j} \neq 0 $, then $ i \in X $ or $ j \in Y $. We then define the crisscross weight of $ E $ as
\begin{equation}
{\rm wt}_c(E) = \min \left\lbrace \# X + \#Y \mid (X,Y) \subseteq [\alpha m] \times [n] \textrm{ is a cover of } E \right\rbrace .
\label{eq def crisscross weight}
\end{equation}
\end{definition}

We may then formalize crisscross error and erasure-correction, together with security, as follows:

\begin{definition} \label{def universal schemes crisscross}
For a subset $ I \subseteq [n] $, define the matrix $ P_I \in \mathbb{F}_q^{\# I \times n} $ as that constituted by the rows of the $ n \times n $ identity matrix indexed by $ I $. We say that the coset coding scheme $ F : \mathbb{F}_q^{\alpha m \times \ell} \longrightarrow \mathbb{F}_q^{\alpha m \times n} $ is:
\begin{enumerate}
\item
Crisscross $ t $-error and $ \rho $-erasure-correcting if, for every $ I \subseteq [n] $ with $ \# I = n - \rho $, there exist a decoding function $ D_I : \mathbb{F}_q^{\alpha m \times (n-\rho)} \longrightarrow \mathbb{F}_q^{\alpha m \times \ell} $ such that
$$ D_I \left( C P_I^T + E \right) = S, $$
for all $ C \in \mathcal{C}_S $, all $ E \in \mathbb{F}_q^{\alpha m \times (n-\rho)} $ with $ {\rm wt}_c(E) \leq t $, and all $ S \in \mathbb{F}_q^{\alpha m \times \ell} $.
\item
Secure under $ \mu $ column-observations if 
$$ H(W \mid S) = H(S), $$
for any matrix $ W \in \mathbb{F}_q^{\alpha m \times \mu} $ constituted by $ \mu $ columns of $ C = F(S) $, for all $ S \in \mathbb{F}_q^{\alpha m \times \ell} $.
\end{enumerate}
\end{definition}

In this scenario, pieces of data correspond to columns, instead of linear combinations of columns, hence we will consider preprocessing functions $ E_{I,i} : \mathbb{F}_q^{\alpha m} \longrightarrow \mathbb{F}_q^{\beta_i m} $ depending on a subset of columns $ I \subseteq [n] $, where $ i \in I $. Hence we may formalize the crisscross error-correction capability of preprocessing functions as follows:

\begin{definition} \label{def preprocessing f correction crisscross}
For a subset $ I \subseteq [n] $ with $ d = \# I $, the preprocessing functions $ E_{I,i} : \mathbb{F}_q^{\alpha m} \longrightarrow \mathbb{F}_q^{\beta_i m} $, for $ i = 1,2, \ldots, d $, are $ t $-crisscross error-correcting with respect to $ F $ if there exists a decoding function $ D_I : \prod_{i = 1}^d \mathbb{F}_q^{\beta_i m} \longrightarrow \mathbb{F}_q^{\alpha m \times \ell} $ such that
\begin{equation} 
 D_I \left( \left( E_{I,i} \left( \mathbf{c}_i + \mathbf{e}_i \right) \right) _{i = 1}^d \right) = S, 
\label{eq condition preprocessing functions 2}
\end{equation}
where $ C \in \mathcal{C}_S $ and $ \mathbf{c}_i $ denotes the $ i $-th column of $ C $, for $ i = 1,2, \ldots, d $, for all $ S \in \mathbb{F}_q^{\alpha m \times \ell} $ and all error matrices $ E \in \mathbb{F}_q^{\alpha m \times d} $ of crisscross weight at most $ t $ with columns $ \mathbf{e}_1, \mathbf{e}_2, \ldots, \mathbf{e}_d \in \mathbb{F}_q^{\alpha m} $.

The decoding bandwidth and communication overhead of such functions are defined as in Definition \ref{def decoding and comm overhead}.
\end{definition}

We now see that the bounds (\ref{eq lower bound info rate}) and (\ref{eq lower bound CO}) also hold in this context:

\begin{proposition}
If the coset coding scheme $ F : \mathbb{F}_q^{\alpha m \times \ell} \longrightarrow \mathbb{F}_q^{\alpha m \times n} $ is crisscross $ t $-error and $ \rho $-erasure-correcting, and secure under $ \mu $ column-observations, then
\begin{equation}
\ell \leq n- 2t - \rho - \mu. \label{eq lower bound info rate crisscross}
\end{equation}
Moreover, for a subset $ I \subseteq [n] $ with $ d = \# I $ and preprocessing functions $ E_{I,i} : \mathbb{F}_q^{\alpha m} \longrightarrow \mathbb{F}_q^{\beta_i m} $, for $ i = 1,2, \ldots, d $, that are $ t $-crisscross error-correcting with respect to $ F $, it holds that:
\begin{equation}
{\rm CO}(I) \geq \frac{\ell (2t + \mu)}{d - 2t - \mu}, \label{eq lower bound CO crisscross}
\end{equation}
\end{proposition}
\begin{proof}
Define $ \rho^\prime = \rho + 2t $. We will prove that $ F $ is $ \rho^\prime $-crisscross erasure-correcting. If it is not, then there exists a subset $ I \subseteq [n] $ with $ \# I = n - \rho^\prime $, and there exist $ C_1 \in \mathcal{C}_{S_1} $ and $ C_2 \in \mathcal{C}_{S_2} $, where $ S_1 \neq S_2 $, such that
$$ C_1P^T_I = C_2P^T_I. $$
Next take a set of the form $ I_1 = I \cup J_1 \cup J_2 $, where $ \# I_1 = n - \rho $ and $ t = \# J_1 = \# J_2 $ (recall that $ n - \rho = n - \rho^\prime + 2t $). There exist matrices $ E_1, E_2 \in \mathbb{F}_q^{\alpha m \times (n - \rho)} $ of crisscross weight at most $ t $ such that
$$ C_1P^T_{I_1} + E_1 = C_2P^T_{I_1} + E_2. $$
Hence $ F $ cannot be crisscross $ t $-error and $ \rho $-erasure-correcting, and we reach a contradiction. Now, this implies that $ F $ is a classical secret sharing scheme with alphabet $ \mathbb{F}_q^{\alpha m} $, reconstruction $ \rho^\prime $ and privacy $ \mu $. Thus it follows directly from \cite[Th. 1]{efficient} that
$$ \ell \leq n - \rho^\prime - \mu = n - \rho - 2t - \mu, $$
and we are done. 

Finally, the bound (\ref{eq lower bound CO crisscross}) can be proven in the same way as the bound (\ref{eq lower bound CO}).
\end{proof}

To conclude, we observe that a coset coding scheme, together with preprocessing functions, which are universally (rank) error and erasure-correcting and universally secure in the sense of Definitions \ref{def universal schemes} and \ref{def preprocessing f correction} are also crisscross erasure and error-correcting and secure under a given number of column observations in the sense of Definitions \ref{def universal schemes crisscross} and \ref{def preprocessing f correction crisscross}, with exactly the same parameters. Thus all constructions in this paper can be directly translated into the context of this subsection.

For illustration purposes, we show how to translate Theorem \ref{theorem using gabidulin} to this context, thus obtaining coset coding schemes which are optimal in the sense of (\ref{eq lower bound info rate crisscross}) and (\ref{eq lower bound CO crisscross}) for all parameters, whenever $ n \leq m $.

\begin{corollary} 
Assume $ n \leq m $, choose integers $ k_2, k_1, t_0 $ and $ \rho_0 $ such that $ 0 \leq k_2 < k_1 \leq n $ and $ 2 t_0 + \rho_0 = n - k_1 $, and choose any subset $ \mathcal{D} \subseteq [n-\rho_0, n] $ with elements $ d_h = n-\rho_0 < d_{h-1} < \ldots < d_2 < d_1 $. 

The coset coding scheme $ F : \mathbb{F}_{q^m}^{\alpha \times \ell} \longrightarrow \mathbb{F}_{q^m}^{\alpha \times n} $ in Theorem \ref{theorem using gabidulin} with these parameters satisfies $ \ell = k_1-k_2 $, is crisscross $ t $-error and $ \rho $-erasure-correcting if $ 2t + \rho \leq n-k_1 $, and is secure under $ \mu $ column-observations if $ \mu \leq k_2 $. In particular, the scheme is optimal in the sense of (\ref{eq lower bound info rate crisscross}). Moreover, it holds that
$$ \alpha = {\rm LCM} \left( d_1 - 2t_0 - k_2, d_2 - 2t_0 - k_2, \ldots, d_h - 2t_0 - k_2 \right). $$

In addition, for any $ d \in \mathcal{D} $ and any subset $ I \subseteq [n] $ with $ d = \# I $, there exist preprocessing functions $ E_{I,i} : \mathbb{F}_q^{\alpha m} \longrightarrow \mathbb{F}_q^{ \ell \alpha m / (d - 2 t_0 - k_2)} $, for $ i = 1,2, \ldots, d $, which are $ t_0 $-crisscross error-correcting and satisfying equality in (\ref{eq lower bound CO crisscross}), hence having optimal communication overheads for all $ d \in \mathcal{D} $.
\end{corollary}

Observe that optimal crisscross error and erasure-correcting coding schemes can also be obtained by using maximum distance separable (MDS) codes in $ \mathbb{F}_q^{\alpha m \times n} $, by identifying this vector space with $ \mathbb{F}_q^{\alpha m n} $, as noticed in \cite{roth}. However, such constructions may require extremely large finite fields, for instance $ q > \alpha m n $ for Reed-Solomon codes, whereas rank-metric codes allow to obtain optimal coding schemes with the only constraint $ n \leq m $, being $ q $ unrestricted, allowing in particular using binary fields ($ q = 2 $).

\section{Conclusion and open problems}

In this paper, we have studied the problem of reducing the communication overhead on a noisy wire-tap channel or storage system where data is encoded as a matrix. The method developed in Section \ref{sec general construction} allows to reduce the communication overhead, when more columns are available, at the cost of expanding the packet length (number of rows). However, in the optimal case of pairs of Gabidulin codes (Section \ref{sec optimal schemes}), strict improvements on the communication overheads are possible at no cost on the rest of the parameters, as shown in Example \ref{example 1} for practical instances in the applications. We leave as open problem to study when the packet length need not be expanded. Another interesting open problem is to extend our method to codes that are linear over the base field $ \mathbb{F}_q $, instead of the extension field $ \mathbb{F}_{q^m} $. This would allow to use all possible MRD codes \cite{delsartebilinear}.

\appendix

\section{Proof of Lemma \ref{lemma secrecy}} \label{app proof of}

Fix $ \mathbb{F}_{q^m} $-linear codes $ \mathcal{C}_2 \varsubsetneqq \mathcal{C}_1 \subseteq \mathbb{F}_{q^m}^n $ and a matrix $ B \in \mathbb{F}_q^{\mu \times n} $ in the rest of the appendix.

We start with an auxiliary result, which is a particular case of \cite[Th. 3]{similarities}:

\begin{lemma} [\textbf{\cite{similarities}}]
It holds that
\begin{equation*}
\begin{split}
d_R(\mathcal{C}_1, \mathcal{C}_2) = \min \{ & {\rm Rk}(A) \mid A \in \mathbb{F}_q^{\nu \times n}, \nu \in \mathbb{N}, \textrm{ and} \\
 & \dim \left( \mathcal{C}_1 \cap {\rm Row}(A) / \mathcal{C}_2 \cap {\rm Row}(A) \right) \geq 1 \},
\end{split}
\end{equation*}
where $ {\rm Row}(A) \subseteq \mathbb{F}_{q^m}^n $ denotes the $ \mathbb{F}_{q^m} $-linear vector space generated by the rows of the matrix $ A \in \mathbb{F}_q^{\nu \times n} $.
\end{lemma}

Given an $ \mathbb{F}_{q^m} $-linear code $ \mathcal{C} \subseteq \mathbb{F}_{q^m}^n $, consider the map $ \mathcal{C} \longrightarrow \mathcal{C} B^T $ defined by $ \mathbf{c} \mapsto \mathbf{c} B^T $, for $ \mathbf{c} \in \mathcal{C} $. It is surjective and its kernel is $ \mathcal{C} \cap \left( \mathcal{V}^\perp \right) $, where $ \mathcal{V} = {\rm Row}(B) $. Therefore
$$ \dim(\mathcal{C}) = \dim \left( \mathcal{C} B^T \right) + \dim \left( \mathcal{C} \cap \left( \mathcal{V}^\perp \right) \right) . $$
Using this equation and computing dimensions, it follows that
\begin{equation}
\dim \left( \mathcal{C}_1 B^T / \mathcal{C}_2 B^T \right) = \dim \left( \mathcal{C}_2^\perp \cap \mathcal{V} / \mathcal{C}_1^\perp \cap \mathcal{V} \right).
\label{eq proof of lemma then equal}
\end{equation}
Now, using that $ {\rm Rk}(B) < d_R \left( \mathcal{C}_2^\perp, \mathcal{C}_1^\perp \right) $ and the previous lemma, it holds that $ \mathcal{C}_2^\perp \cap \mathcal{V} = \mathcal{C}_1^\perp \cap \mathcal{V} $. Hence the result follows by (\ref{eq proof of lemma then equal}).

\section*{Acknowledgement}

The author gratefully acknowledges the support from The Danish Council for Independent Research (Grant No. DFF-4002-00367 and Grant No. DFF-5137-00076B ``EliteForsk-Rejsestipendium''), and is thankful for the guidance of his advisors Olav Geil and Diego Ruano. This manuscript was written in part when the author was visiting the University of Toronto. He greatly appreciates the support and hospitality of Frank R. Kschischang, and is thankful for valuable discussions on this work.

%\nocite*{}
\bibliographystyle{IEEEtranS}
%\bibliography{rankcomm}

% Generated by IEEEtranS.bst, version: 1.14 (2015/08/26)

\end{document}